\documentclass[reqno, 11pt]{amsart}
\usepackage{amsmath,amsthm}
\usepackage{amsfonts,amssymb}

\usepackage{enumerate}

\newtheorem{Theorem}{Theorem}[section]

\newtheorem{Proposition}[Theorem]{Proposition}
\newtheorem{Corollary}[Theorem]{Corollary}

\theoremstyle{definition}
\newtheorem{Definition}[Theorem]{Definition}
\newtheorem{Remark}[Theorem]{Remark}

\newcommand{\thref}[1]{Theorem \ref{#1}}

\newcommand{\prref}[1]{Proposition \ref{#1}}
\newcommand{\reref}[1]{Remark \ref{#1}}

\newcommand{\coref}[1]{Corollary \ref{#1}}
\newcommand{\seref}[1]{Section \ref{#1}}

\numberwithin{equation}{section}

\let\emptyset\varnothing

\newcommand{\Rset}{\mathbb{R}}

\newcommand{\Zset}{\mathbb{Z}}
\newcommand{\Nset}{\mathbb{N}}
\newcommand{\Pset}{\mathbb{P}}
\newcommand{\Sset}{\mathbb{S}}
\newcommand{\Tset}{\mathbb{T}}

\newcommand{\al}{{\alpha}}
\newcommand{\be}{{\beta}}
\newcommand{\la}{{\lambda}}
\newcommand{\ka}{{\kappa}}

\newcommand{\mut}{{\tilde{\mu}}}

\newcommand{\mub}{{\bar{\mu}}}

\newcommand{\nuh}{{\hat{\nu}}}

\newcommand{\beh}{{\hat{\be}}}
\newcommand{\bet}{{\tilde{\be}}}
\newcommand{\Ga}{{\Gamma}}
\newcommand{\ga}{{\gamma}}

\newcommand{\si}{{\sigma}}

\newcommand{\pd}{\partial}

 \newcommand{\bv}{{\mathbf v}}
 \newcommand{\bx}{{\mathbf x}}
 
\newcommand{\bz}{{\mathbf z}}

\newcommand{\bga}{\boldsymbol{\gamma}}
\newcommand{\bnu}{\boldsymbol{\nu}}
\newcommand{\bmu}{\boldsymbol{\mu}}
\newcommand{\bbe}{\boldsymbol{\beta}}

\newcommand{\cG}{{\mathcal G}}
\newcommand{\cL}{{\mathcal L}}
\newcommand{\cN}{{\mathcal N}}
\newcommand{\cM}{{\mathcal M}}
\newcommand{\cP}{{\mathcal P}}
\newcommand{\cR}{{\mathcal R}}
\newcommand{\cS}{{\mathcal S}}
\newcommand{\cSt}{\tilde{{\mathcal S}}}
\newcommand{\cH}{{\mathcal H}}
\newcommand{\cB}{{\mathcal B}}
\newcommand{\cI}{{\mathcal I}}

\newcommand{\td}{\mathfrak{t}_{d+1}}
\newcommand{\sld}{\mathfrak{sl}_{d+1}}

\newcommand{\fGd}{{\mathfrak G}_{d+1}}

\newcommand{\cDz}{{\mathcal D}_z}

\newcommand{\Ph}{\hat{P}}
\newcommand{\Rh}{\hat{R}}

\begin{document}

\title[Quantum system on the sphere and Racah operators]{The generic quantum superintegrable system on the sphere and Racah operators}
\keywords{quantum superintegrable systems, symmetries, commuting operators, classical orthogonal polynomials, bispectrality}
\subjclass[2010]{81R12, 17B81, 33C80}

\date{June 12, 2017}

\author{Plamen Iliev}
\address{School of Mathematics, Georgia Institute of Technology, Atlanta, GA 30332--0160, USA}
\email{iliev@math.gatech.edu}
\thanks{The author is partially supported by Simons Foundation Grant \#280940.}

\begin{abstract}
We consider the generic quantum superintegrable system on the $d$-sphere with potential 
$V(y)=\sum_{k=1}^{d+1}\frac{b_k}{y_k^2}$, where  $b_k$ are parameters. Appropriately normalized, the symmetry operators for the Hamiltonian define a representation of the Kohno-Drinfeld Lie algebra on the space of polynomials orthogonal with respect to the Dirichlet distribution. The Gaudin subalgebras generated by Jucys-Murphy elements are diagonalized by families of Jacobi polynomials in $d$ variables on the simplex. We define a set of generators for the symmetry algebra and we prove that their action on the Jacobi polynomials is represented by the multivariable Racah operators introduced in \cite{GI}. 
The constructions also yield a new Lie-theoretic interpretation of the bispectral property for Tratnik's multivariable Racah polynomials.
\end{abstract}
\maketitle
 
\section{Introduction}\label{se1}
Let
$\Sset^{d}=\{y\in\Rset^{d+1}:y_1^2+\cdots+y_{d+1}^2=1\}$
denote the $d$-dimensional sphere,
and let $\cH$ denote the quantum Hamiltonian on $\Sset^{d}$ defined by
\begin{equation}\label{1.1}
\cH=\sum_{1\leq i<j \leq d+1}\left(y_i{\pd_{y_j}}-y_j{\pd_{y_i}}\right)^2+\sum_{k=1}^{d+1}\frac{b_k}{y_k^2}, \qquad\text{ where }
\qquad\pd_{y_j}=\frac{\pd}{\pd {y_j}},
\end{equation}
and $\{b_k\}_{k=1,\dots,d+1}$ are parameters. It is easy to check that the operators 
\begin{equation}\label{1.2}
L_{i,j}=\left(y_i{\pd_{y_j}}-y_j{\pd_{y_i}}\right)^2+\frac{b_iy_j^2}{y_i^2}+\frac{b_jy_i^2}{y_j^2}
\end{equation}
commute with $\cH$ and therefore generate a symmetry algebra. 
It is not hard to see that the system is completely integrable since it admits $d$ algebraically independent and mutually commuting operators $\cI_1=\cH$, $\cI_2$,\dots, $\cI_d$, where $\cI_j=\sum_{k=j+1}^{d+1}L_{j,k}$ for $j\geq 2$ (see \reref{re2.2}).  This system has been extensively studied in the literature \cite{KMP1,KMP2,KMT,MPW,MT} as an important example of a second-order superintegrable system, possessing $(2d-1)$ second-order algebraically independent symmetries. We refer to this system as the {\em generic superintegrable system on the sphere}, following the terminology used in dimensions $2$ and $3$.

In a series of papers \cite{KMP1,KMP2}, Kalnins, Miller and Post described the irreducible representations of the symmetry algebra of the Hamiltonian $\cH$ in dimensions $d=2$ and $d=3$ and discovered an interesting link to Racah polynomials and their two-variable extensions proposed by Tratnik \cite{Tr}, respectively. For the $3$-sphere, they noticed that the action of appropriate linear combinations of the generators $L_{i,j}$ of the symmetry algebra can be expressed in terms of the two-dimensional Racah operators constructed in \cite{GI} and raised the natural question whether this phenomenon extends in higher dimensions. The $3$-dimensional case was further analyzed recently by Post \cite{P}, building on the work by Genest and Vinet \cite{GV}.

The goal of the present paper is to extend the connection between the symmetry algebra for the Hamiltonian \eqref{1.1} and the Racah operators defined in \cite{GI} in arbitrary dimension, by generalizing the constructions in \cite{I1}. There, we used a representation of the Lie algebra $\sld$ together with two Cartan subalgebras, which are now replaced by  the Kohno-Drinfeld Lie algebra together with two Gaudin subalgebras. First, we note that, appropriately normalized, the symmetry algebra for the Hamiltonian $\cH$ defines a representation of the Kohno-Drinfeld Lie algebra on the space of polynomials orthogonal with respect to the Dirichlet distribution \cite{KMT}. The Gaudin subalgebras generated by Jucys-Murphy elements are then diagonalized by families of Jacobi polynomials in $d$ variables on the simplex. We fix one such Gaudin subalgebra $\fGd$ corresponding to the standard basis $P_{\nu}$ of polynomials and we define a second Gaudin subalgebra $\fGd^{\tau}$ and a second basis $P^{\tau}_{\nu}$ of polynomials by applying appropriately the cyclic permutation $\tau=(1,2,\dots,d+1)$ to $\fGd$ and $P_{\nu}$, respectively. We prove that the action of the Gaudin algebras $\fGd$ and $\fGd^{\tau}$ on each of the bases $\{P_{\nu}\}$ and $\{P^{\tau}_{\nu}\}$ can be written in terms of the $(d-1)$-dimensional Racah algebras of operators and variables defined in \cite{GI}. In particular, if we fix the basis $\{P_{\nu}\}$, then we obtain 
explicit formulas for the action of the operators in the algebras $\fGd$, $\fGd^{\tau}$ and $\fGd^{\tau^{-1}}=\tau^{-1}\circ\fGd$. In dimensions $d=2$ and $d=3$, linear combinations of the operators in $\fGd$, $\fGd^{\tau}$ and $\fGd^{\tau^{-1}}$ lead to the generators $L_{i,j}$ of the symmetry algebra and correspond to the formulas obtained in \cite{KMP1,KMP2,P}. In dimension $d>3$, we show that the operators in $\fGd$, $\fGd^{\tau}$ and $\fGd^{\tau^{-1}}$ still generate the full symmetry algebra, but we need to use nonlinear relations. An important ingredient in the proof stems from a result established recently in \cite{IX} which allows to express the transition matrix from $P_{\nu}$ to $P^{\tau}_{\nu}$ in terms of the Racah polynomials in $(d-1)$ variables defined by Tratnik \cite{Tr}. 
As an immediate corollary of these constructions, we also obtain a new Lie-theoretic interpretation of the bispectral property established in \cite{GI}.

It is perhaps useful to stress that the complexity increases exponentially as we move to higher dimensions. For instance, for the $4$-sphere, the corresponding Racah algebra contains $3$ mutually commuting difference operators, one of them having $27$ (rather complex) coefficients. If we go to dimension $d=5$, we will have an operator with $81$ coefficients, etc.

The paper is organized as follows. In the next section we normalize the symmetry operators, so that they act naturally on the space of polynomials and we exhibit a set of $(2d-1)$ operators which generate the symmetry algebra. In \seref{se3}, we explain the relation to Gaudin subalgebras and Jacobi polynomials. In \seref{se4}, we provide a short introduction to the multivariable Racah polynomials and operators. In \seref{se5}, we prove the main results and discuss their connection to the bispectral problem.

\section{Symmetry algebra}\label{se2}
First, we normalize the symmetry operators for the Hamiltonian \eqref{1.1}. 
If we consider nonnegative coordinates $y_i\geq 0$ and set $z_i=y_i^2$, then the operators $L_{i,j}$ take the form
\begin{equation*}
L_{i,j}=4z_iz_j\left({\pd_{z_j}}-{\pd_{z_i}}\right)^2+2(z_j-z_i)\left({\pd_{z_i}}-{\pd_{z_j}}\right)+\frac{b_iz_j}{z_i}+\frac{b_jz_i}{z_j},
\end{equation*}
on the simplex $\{z\in\Rset^{d+1}:z_1+\cdots+z_{d+1}=1\text{ and } z_i\geq 0\}$.
Furthermore, if we consider the gauge factor
$$\cG_{\al}(z)=\prod_{j=1}^{d+1}z_j^{\al_j}$$
then a straightforward computation shows that 
\begin{equation*}
\begin{split}
\cG_{\al}(z)L_{i,j}\circ \cG^{-1}_{\al}(z)=&4z_iz_j\left({\pd_{z_j}}-{\pd_{z_i}}\right)^2+4\left((2\al_j-1/2)z_i-(2\al_i-1/2)z_j\right)\left({\pd_{z_i}}-{\pd_{z_j}}\right)\\
&+\left((2\al_i+1/2)^2+b_i-1/4\right)\frac{z_j}{z_i}+\left((2\al_j+1/2)^2+b_j-1/4\right)\frac{z_i}{z_j}\\
&+2[(\al_i+\al_j)-4\al_i\al_j].
\end{split}
\end{equation*}
Thus, if we set $\ga_i=-(2\al_i+1/2)$ the last equation can be rewritten as
\begin{equation*}
\begin{split}
\cG_{\al}(z)L_{i,j}\circ \cG^{-1}_{\al}(z)=&4z_iz_j\left({\pd_{z_j}}-{\pd_{z_i}}\right)^2+4\left[-(\ga_j+1)z_i+(\ga_i+1)z_j\right] \left({\pd_{z_i}}-{\pd_{z_j}}\right)\\
&+\left(\ga_i^2+b_i-1/4\right)\frac{z_j}{z_i}+\left(\ga_j^2+b_j-1/4\right)\frac{z_i}{z_j}\\
&-2[(\ga_i+1)(\ga_j+1)-1/4].
\end{split}
\end{equation*}
In particular, if we replace the parameters $\{b_i\}_{i=1,\dots,d+1}$ with parameters $\{\ga_i\}_{i=1,\dots,d+1}$, related by 
$$b_i=\frac{1}{4}-\ga_i^2,$$
then
\begin{equation*}
\cG_{\al}(z)L_{i,j}\circ \cG^{-1}_{\al}(z)=4t_{i,j}-2[(\ga_i+1)(\ga_j+1)-1/4],
\end{equation*}
where 
$$t_{i,j}=z_iz_j\left({\pd_{z_j}}-{\pd_{z_i}}\right)^2+\left[(\ga_i+1)z_j-(\ga_j+1)z_i\right] \left({\pd_{z_i}}-{\pd_{z_j}}\right).$$
Thus, up to a gauge transformation and unessential constant factors, we can replace the symmetry operators by the operators $t_{i,j}$. Finally, if we choose coordinates $x_1=z_1$, $x_2=z_2$, \dots, $x_d=z_d$, the operators $t_{i,j}$ take the form
\begin{equation}\label{2.1}
\begin{split}
&t_{i,j}=x_ix_j(\pd_{x_i}-\pd_{x_j})^2+[(\ga_i+1)x_j-(\ga_j+1)x_i](\pd_{x_i}-\pd_{x_j}), \\
&\hskip7cm  \text{ if }i\neq j\in\{1,\dots,d\},\\
&t_{j,d+1}=t_{d+1,j}=x_j(1-|x|)\pd_{x_j}^2+[(\ga_j+1)(1-|x|)-(\ga_{d+1}+1)x_j]\pd_{x_j},\\
&\hskip7cm  \text{ if }j\in\{1,\dots,d\},
\end{split}
\end{equation}
where $|x|=x_1+x_2+\cdots+x_d$.  The above computations are similar to the ones in \cite{KMT}, where the starting point was the second-order partial differential operator $\sum_{1\leq i<j\leq d+1}t_{i,j}$ for the Lauricalla functions.

\begin{Definition}\label{de2.1}
We denote by  $\td$ the associative algebra generated by the operators $t_{i,j}$, $i\neq j\in \{1,\dots,d+1\}$ defined in \eqref{2.1}.
\end{Definition}
From now on, we focus on the algebra $\td$ and the operators $t_{i,j}$ defined in \eqref{2.1}. 
As we noted above, up to a gauge transformation and unessential constant terms, $\td$ coincides with the symmetry algebra for the quantum Hamiltonian $\cH$ in \eqref{1.1}. 

\begin{Remark}\label{re2.2}
It easy to check that the operators $t_{i,j}$ satisfy the following commutation relations
\begin{subequations}\label{2.2}
\begin{align}
[t_{i,j},t_{k,l}]&=0, \text{ if }i,j,k,l \text{ are distinct,}\label{2.2a}\\
[t_{i,j},t_{i,k}+t_{j,k}]&=0, \text{ if }i,j,k \text{ are distinct.} \label{2.2b}
\end{align}
\end{subequations}
Recall that the Kohno-Drinfeld Lie algebra \cite{Dr,Ko} is the quotient of the free Lie algebra on generators $t_{i,j}$, by the ideal generated by the relations in \eqref{2.2}. Thus, the symmetry operators $t_{i,j}$ in \eqref{2.1} define a representation of the Kohno-Drinfeld Lie algebra. In particular, if $S_{d+1}$ is the symmetric group consisting of all permutations of $d+1$ symbols, then for every permutation $\sigma\in S_{d+1}$, the Jucys-Murphy elements 
$$t_{\si_1,\si_2}, t_{\si_1,\si_3}+t_{\si_2,\si_3}, t_{\si_1,\si_4}+t_{\si_2,\si_4}+t_{\si_3,\si_4},\dots, \sum_{j=1}^{d}t_{\si_j,\si_{d+1}}$$ 
commute with each other and generate a commutative subalgebra of $\td$. We will refer to this subalgebra as a Gaudin subalgebra of $\td$, following the convention in the literature \cite{AFV,Fr}.
\end{Remark}

The operators $t_{i,j}$ defined in \eqref{2.1} satisfy other rather complicated relations. In dimensions $d=2$ and $d=3$, the structure equations can be found in the works of Kalnins, Miller and Post \cite{KMP1,KMP2}. Note that when $d$ increases, the dimension of the space of second-order symmetries $\binom{d+1}{2}$ grows quadratically in $d$, while the number of the algebraically independent symmetries $2d-1$  is a linear function of $d$. Thus, for many practical purposes, it is crucial to find a smaller explicit set of generators for the symmetry algebra. It turns out that a single fourth-order relation can be used to reduce the set of generators. If $i,j,k,l$ are distinct indices, one can show that 
\begin{align}
&(1-\ga_k^2)(1-\ga_l^2)t_{i,j}=
\left \{ [t_{j,k},t_{k,l}],[t_{i,k},t_{k,l}] \right\}-2\left\{t_{k,l},t_{i,k}t_{j,l}\right\}\nonumber\\
&\qquad-\left\{t_{k,l},[t_{i,k},[t_{j,k},t_{k,l}]] \right\}+(1+\ga_k)(1+\ga_l)\left[t_{i,k},[t_{k,l},t_{j,l}]\right]\nonumber\\
&\qquad+(1+\ga_j)(1+\ga_l)\left(\left\{t_{i,k},t_{k,l}\right\}-2\ga_k t_{i,k}-(1+\ga_i)(1+\ga_k)t_{k,l} \right)\nonumber\\
&\qquad+(1-\ga_l^2)\left\{t_{i,k},t_{j,k}\right\}+(1-\ga_k^2)\left\{t_{i,l},t_{j,l}\right\}+(1+\ga_i)(1+\ga_k) \left\{t_{j,l},t_{k,l}\right\}\nonumber\\
&\qquad-4t_{j,k}t_{i,l}+2(-1+\ga_k+\ga_l+\ga_k\ga_l)t_{j,l}t_{i,k}\nonumber\\
&\qquad+(1+\ga_i)(1+\ga_l)(1-\ga_k+\ga_l+\ga_k\ga_l)t_{j,k}\nonumber\\
&\qquad-2(1+\ga_i)(1+\ga_k)\ga_lt_{j,l}+(1+\ga_j)(1+\ga_k)(1+\ga_k-\ga_l+\ga_k\ga_l)t_{i,l}.\label{2.3}
\end{align}
Here, as usual, $\{A,B\}=AB+BA$ denotes the anticommutator of the operators $A$ and $B$.
Note that the right-hand side of the last equation is generated by the elements  $t_{i,k}$, $t_{i,l}$, $t_{j,k}$, $t_{j,l}$, $t_{k,l}$. We assume throughout the paper that $\ga_s\neq \pm1$ and therefore we deduce from the above formula that $t_{i,j}$ is generated by these elements:
\begin{equation}\label{2.4}
t_{i,j}\in\Rset \langle t_{i,k}, t_{i,l}, t_{j,k}, t_{j,l}, t_{k,l}\rangle.
\end{equation}
As an immediate corollary, we obtain an explicit set of $(2d-1)$ generators for $\td$.
\begin{Proposition}\label{pr2.3}
The algebra $\td$ is generated by the set 
\begin{equation}\label{2.5}
\cS=\{t_{1,j}:j=2,3,\dots,d+1\}\cup \{t_{i,d+1}:i=2,3,\dots,d\}.
\end{equation}
\end{Proposition}
\begin{proof}
The statement is obvious when $d=2$. When $d\geq 3$ and $1<i<j<d+1$ we can generate $t_{i,j}$ using the elements $t_{1,i}$, $t_{1,j}$, $t_{i,d+1}$, $t_{j,d+1}$, $t_{1,d+1}$ from $\cS$.
\end{proof}

\section{Gaudin subalgebras of  $\td$ and Jacobi polynomials}\label{se3}
For a vector $v=(v_1,\dots,v_s)$ we denote by $|v|=v_1+\cdots+v_s$ the sum of its components. 
Suppose now that $\ga=(\ga_1,\dots,\ga_{d+1})$ is such that $\ga_j>-1$ for all $j\in\{1,\dots,d+1\}$. For $x=(x_1,\dots,x_d)$ let
\begin{equation}\label{3.1}
W_{\ga} (x)= \frac{\Ga(|\ga|+d+1)}{\prod_{j=1}^{d+1}\Ga(\ga_j+1)}\, x_1^{\ga_1} \cdots x_d^{\ga_d} (1-|x|)^{\ga_{d+1}},  
\end{equation}
denote the Dirichlet distribution on the simplex 
$$\Tset^{d}=\{x\in\Rset^{d}:  x_i\geq 0\text{ and }|x|\leq 1\}.$$
On the space $\Rset[x]$ of polynomials of $x_1,x_2,\dots,x_d$, define an inner product by
\begin{equation}\label{3.2}
\langle f, g\rangle =\int_{\Tset^d}f(x)g(x) W_{\ga} (x)\, dx.
\end{equation}
Let $\Pset_n$ be the space of polynomials of total degree at most $n$ with the convention $\Pset_{-1}=\{0\}$. The algebra $\td$ defined in the previous section has a natural action on the space of orthogonal polynomials with respect to the inner product \eqref{3.2}.
\begin{Proposition}\label{pr3.1}
Let $i\neq j\in \{1,\dots,d+1\}$. Then
\begin{enumerate}[{\rm (i)}]
\item The operator $t_{i,j}$ is self-adjoint with respect to the inner product \eqref{3.2}.
\item For $n\in\Nset_0$ we have $t_{i,j}(\Pset_n)\subset \Pset_n$, i.e. $t_{i,j}:\Pset_n\to\Pset_n$.
\end{enumerate}
\end{Proposition}
\begin{proof}
If $i\neq j\in \{1,\dots,d\}$ it is easy to see that 
$$t_{i,j}=\frac{1}{W_{\ga} (x)}\,  (\pd_{x_i}-\pd_{x_j})\,x_ix_j W_{\ga} (x) (\pd_{x_i}-\pd_{x_j}).$$
Using this and integrating by parts, it follows that
$$\langle t_{i,j}f, g\rangle =-\int_{\Tset^d}x_ix_j [ (\pd_{x_i}-\pd_{x_j}) f(x)] \, [ (\pd_{x_i}-\pd_{x_j}) g(x)]\, W_{\ga} (x)\, dx.$$
Since the right-hand side is symmetric in $f$ and $g$, we deduce that $\langle t_{i,j}f, g\rangle =\langle f,  t_{i,j}g\rangle$. 
Similarly, using the representation 
$$t_{i,d+1}=\frac{1}{W_{\ga} (x)}\,  \pd_{x_i}\,x_i (1-|x|) W_{\ga} (x) \pd_{x_i}$$
it follows that $t_{i,d+1}$ is self-adjoint with respect to the inner product \eqref{3.2} thus completing the proof of (i). The proof of (ii) is straightforward.
\end{proof}

For $n\in\Nset_0$, let $\cP^{\ga}_n=\Pset_n\ominus\Pset_{n-1}$ denote the space of polynomials of total degree $n$, orthogonal to all polynomials of total degree at most $n-1$ with respect to the inner product \eqref{3.2}. Then, as an immediate consequence of \prref{pr3.1}, we obtain the following corollary.
\begin{Corollary}\label{co3.2}
For $n\in\Nset_0$ and $i\neq j\in\{1,\dots,d+1\}$ we have 
$$t_{i,j}(\cP^{\ga}_n)\subset \cP^{\ga}_n,$$
i.e. we can restrict the representation of $\td$ onto the space $\cP^{\ga}_n$.
\end{Corollary}

The Gaudin subalgebras generated by Jucys-Murphy elements in $\td$ define mutually orthogonal bases of orthogonal polynomials which are products of ${}_2F_1$ hypergeometric functions. 
In order to write explicit formulas, we will introduce some notations. For a vector $v=(v_1,\dots, v_{s})$ we define 
\begin{equation*}
    \bv_j = (v_1, \dots, v_j) \quad \text{ and }\quad \bv^j = (v_j, \ldots, v_s), 
\end{equation*}
with the convention that $\bv_0 = \emptyset$ and $\bv^{s+1} =  \emptyset$. We also use standard multi-index notation throughout the paper. For instance, if $\nu=(\nu_1,\dots,\nu_d)\in\Nset_0^d$ then
$$x^{\nu}=x_1^{\nu_1}\cdots x_d^{\nu_d} \qquad\text{and}\qquad \nu!=\nu_1!\cdots \nu_d!.$$
We will use the Jacobi polynomial $p_n^{(\al,\be)}$ normalized as follows
\begin{equation*}
p_n^{(\al,\be)}(t) =  \frac{(\al+1)_n}{(\be+1)_n} {}_2F_1 \left (\begin{matrix} -n, n+\al+\be+1 \\ \al+1 \end{matrix}; \frac{1-t}2 \right).
\end{equation*}
For $\nu \in \Nset_0^d$  we define
\begin{equation}\label{3.3}
a_j=a_j(\ga,\nu)=|\bga^{j+1}| + 2 |\bnu^{j+1}| + d-j, \qquad 1 \le j \le d.
\end{equation} 

With these notations, an orthogonal basis of $\Rset[x]$ for the inner product \eqref{3.2}
is given by
\begin{equation}\label{3.4}
P_{\nu}(x;\ga) = \prod_{k=1}^d \left(1-|\bx_{k-1}| \right)^{\nu_k} 
               p_{\nu_k}^{(a_k,\ga_k)}\left (\frac{2x_k}{1-|\bx_{k-1}|} -1\right), 
\end{equation}
with norms
\begin{equation}\label{3.5}
||P_{\nu}||^2 =  \langle P_\nu, P_\nu \rangle=
   \frac{1}{(|\ga|+d+1)_{2|\nu|}}
     \prod_{j=1}^d \frac{(\ga_j+a_j+\nu_j+1)_{\nu_j}(a_j+1)_{\nu_j}\nu_j!} {(\ga_j +1)_{\nu_j} },
\end{equation}
see \cite[p. 150]{DX}.

The basis $\{P_\nu(x;\ga)\}_{\nu\in\Nset_0^d}$ can be characterized by the fact that it diagonalizes the Gaudin subalgebra $\fGd$ of $\td$ defined by
\begin{equation}\label{3.6}
\fGd=\Rset\left\langle t_{d,d+1}, t_{d-1,d}+t_{d-1,d+1}, t_{d-2,d-1}+t_{d-2,d}+t_{d-2,d+1},\dots,\sum_{j=2}^{d+1}t_{1,j}\right\rangle.
\end{equation}
This can be deduced from the results in \cite{KMT} and \cite[Section 5.3]{GI}, but we provide a short direct proof below. 
Define operators 
\begin{equation}\label{3.7}
\cM_{j,d}(x)=\sum_{j\leq k<l\leq d+1}t_{k,l}, \qquad \text{ for }\qquad j=1,2,\dots,d,
\end{equation}
and note that 
$$\fGd=\Rset\langle\cM_{1,d}(x),\dots,\cM_{d,d}(x)\rangle.$$ 
In the rest of the paper, we will write simply $\cM_j$ for $\cM_{j,d}(x)$ when the variables $x$ and the dimension $d$ are fixed. 
With these notations, the following spectral equations hold.
\begin{Proposition}\label{pr3.3}
For $j=1,\dots,d$ we have
\begin{equation}\label{3.8}
\cM_{j}P_\nu (x;\ga)=-|\bnu^{j}|(|\bnu^{j}|+|\bga^{j}|+d+1-j)P_\nu (x;\ga).
\end{equation}
\end{Proposition}

\begin{proof}
First, note that the operator $\cM_{1,d}(x)=\sum_{1\leq k<l\leq d+1}t_{k,l}$ can be written as 
$$\cM_{1,d}(x)=\sum_{i=1}^{d}x_i(1-x_i)\pd_{x_i}^2-2\sum_{1\leq i<j\leq d}x_ix_j\pd_{x_i}\pd_{x_j}+\sum_{i=1}^{d}\left(\ga_i+1-(|\ga|+d+1)x_i\right)\pd_{x_i}.$$
From this formula, it is easy to see that $\cM_{1,d}(x)$ has a triangular action on $\Rset[x]$ with respect to the total degree as follows
\begin{equation}\label{3.9}
\cM_{1,d}(x)x^{\nu}=-|\nu|(|\nu|+|\ga|+d)x^{\nu}\mod \Pset_{|\nu|-1}.
\end{equation}
This combined with \prref{pr3.1} shows that
\begin{equation}\label{3.10}
\cM_{1,d}(x)P_\nu(x;\ga)=-|\nu|(|\nu|+|\ga|+d)P_\nu(x;\ga),
\end{equation}
thus establishing equation \eqref{3.8} when $j=1$. Fix now $j>1$ and note that the operator $\cM_{j,d}(x)$ contains no derivatives with respect to $x_1,\dots,x_{j-1}$. Therefore, the first $(j-1)$ terms in the product in equation \eqref{3.4} commute with $\cM_{j,d}(x)$.  If we introduce new variables $y_j,y_{j+1},\dots,y_{d}$ by 
\begin{equation*}
y_k=\frac{x_k}{1-|\bx_{j-1}|},\qquad\text{ for }\qquad k=j,j+1,\dots,d,
\end{equation*}
then one can check that 
\begin{equation*}
\cM_{j,d}(x)=\cM_{1,d+1-j}(y),
\end{equation*}
and 
\begin{equation*}
\prod_{k=j}^d \left(1-|\bx_{k-1}| \right)^{\nu_k} 
               p_{\nu_k}^{(a_k,\ga_k)}\left (\frac{2x_k}{1-|\bx_{k-1}|} -1\right)=P_{\bnu^j} (y;\bga^{j}).
\end{equation*}
The proof of \eqref{3.8} now follows from \eqref{3.10}.
\end{proof}

\section{Racah operators}\label{se4}
Consider variables $z_1,z_2,\dots$ and parameters $\be_0,\be_1,\dots$. We work below with functions and operators involving only a finite number of the variables $z_i$ and the parameters $\be_j$, but it will be convenient to use semi-infinite vectors by setting $z=(z_1,z_2,\dots)$ and $\be=(\be_0,\be_1,\dots)$. Extending the convention in the previous section for $j\in\Nset$ we have $\bz_j=(z_1,\dots,z_j)$ and $\bbe_j=(\be_0,\be_1,\dots,\be_j)$. From now on, we adopt the convention that any finite-dimensional vector can also be considered as a semi-infinite vector by adding zeros after the last component.

We denote by $\Rset(z)$ the field of rational functions of finitely many of the $z_j$'s and for $k\in\Nset$ we define an involution $I_k$ on $\Rset(z)$, by 
\begin{equation}\label{4.1}
I_k(z_k)=-z_k-\be_k \text{ and }I_k(z_j)=z_j \text{ for }j\neq k.
\end{equation}
For $k\in\Nset$ we denote by $E_{z_k}$ the forward shift operator acting on the variable $z_k$. Explicitly, if $f(z)\in\Rset(z)$ then 
\begin{equation*}
E_{z_k}f(z_1,z_2,\dots,z_{k-1},z_{k},z_{k+1},\dots)=f(z_1,z_2,\dots,z_{k-1},z_{k}+1,z_{k+1},\dots),
\end{equation*}
and its inverse  $E_{z_k}^{-1}$ corresponds to the backward shift in the variable $z_k$:
\begin{equation*}
E_{z_k}^{-1}f(z_1,z_2,\dots,z_{k-1},z_{k},z_{k+1},\dots)=f(z_1,z_2,\dots,z_{k-1},z_{k}-1,z_{k+1},\dots).
\end{equation*}
Let $\Zset^{\infty}=\{(\nu_1,\nu_2,\dots): \nu_j\neq 0 \text{ for finitely many }j\}$ be the additive group consisting of semi-infinite vectors having finitely many nonzero integer entries. Note that for $\nu\in\Zset^{\infty}$ we have a well-defined shift operator
$$E_z^{\nu}=E_{z_1}^{\nu_1}E_{z_2}^{\nu_2}E_{z_3}^{\nu_3}\cdots,$$
since the right-hand side has only finitely many terms different from the identity operator.
We denote by $\cDz$ the associative algebra of difference operators of the form 
$$L=\sum_{\nu\in S}l_\nu(z)E_z^{\nu},$$
where $S$ is a finite subset of $\Zset^{\infty}$ and $l_{\nu}(z)\in\Rset(z)$. The involution $I_k$ can be extended to an involution on $\cDz$ by defining 
\begin{equation}\label{4.2}
I_k(E_{z_k})=E_{z_k}^{-1} \text{ and }I_k(E_{z_j})=E_{z_j} \text{ for }j\neq k.
\end{equation}
We say that an operator $L\in\cDz$ is $I$-invariant, if it is invariant under the action of all involutions $I_k$, $k\in\Nset$. 

Next, we define a commutative subalgebra of $\cDz$ consisting of $I$-invariant operators, which will refer to as {\em the Racah operators}. 
For $i\in\Nset_0$ and $(j,k)\in\{0,1\}^2$ we define $B_{i}^{j,k}$
as follows
\begin{subequations}\label{4.3}
\begin{align}
B_i^{0,0}&=z_{i}(z_{i}+\be_{i})+z_{i+1}(z_{i+1}+\be_{i+1})+\frac{(\be_{i}+1)(\be_{i+1}-1)}{2},\label{4.3a}\\
B_i^{0,1}&=(z_{i+1}+z_{i}+\be_{i+1})(z_{i+1}-z_{i}+\be_{i+1}-\be_{i}),\label{4.3b}\\
B_i^{1,0}&=(z_{i+1}-z_{i})(z_{i+1}+z_{i}+\be_{i+1}),\label{4.3c}\\
B_i^{1,1}&=(z_{i+1}+z_{i}+\be_{i+1})(z_{i+1}+z_{i}+\be_{i+1}+1), \label{4.3d}
\end{align}
\end{subequations}
where $z_0=0$.
For $i\in\Nset$ we denote
\begin{subequations}\label{4.4}
\begin{align}
b_i^{0}&=\frac{(2z_{i}+\be_i+1)(2z_{i}+\be_i-1)}{2},\label{4.4a}\\
b_i^{1}&=(2z_{i}+\be_i+1)(2z_{i}+\be_i).\label{4.4b}
\end{align}
\end{subequations}
Using the above notations, for $j\in\Nset$ and $\nu\in\{0,1\}^j$ we define 
\begin{subequations}\label{4.5}
\begin{equation}\label{4.5a}
C_{j,\nu}(z)=\frac{\prod_{k=0}^{j}B_k^{{\nu}_k,{\nu}_{k+1}}}
{\prod_{k=1}^{j}b_k^{{\nu}_k}}, 
\end{equation}
where $\nu_0=\nu_{j+1}=0$. We extend the definition of $C_{j,\nu}$ for $\nu\in\{-1,0,1\}^j$ using the involutions $I_k$ as follows. 
Every $\nu\in\{-1,0,1\}^j$ can be decomposed as $\nu=\nu^{+}-\nu^{-}$, 
where $\nu^{\pm}\in \{0,1\}^j$ with components $\nu_k^{+}=\max(\nu_k,0)$ and 
$\nu_k^{-}=-\min(\nu_k,0)$.
For $\nu\in\{-1,0,1\}^j\setminus \{0,1\}^j$ we 
define 
\begin{equation}\label{4.5b}
C_{j,\nu}(z)=I^{\nu^{-}}(C_{\nu^{+}+\nu^{-}}(z)),
\end{equation}
\end{subequations}
where $I^{\nu^{-}}$ is the composition of the involutions corresponding to the positive coordinates of $\nu^{-}$.
Finally, for $j\in\Nset$ we define 
\begin{equation}\label{4.6}
\cL_j(z;\be)=\sum_{\nu\in\{-1,0,1\}^j}C_{j,\nu}(z)E_z^{\nu}
-\left(z_{j+1}(z_{j+1}+\be_{j+1})+\frac{(\be_{0}+1)(\be_{j+1}-1)}{2}\right).
\end{equation}
Note that $\cL_j(z;\be)$ is an $I$-invariant difference operator in the variables $\bz_{j}=(z_1,\dots,z_{j})$ 
with coefficients depending rationally on $\bz_{j+1}=(z_1,\dots,z_{j+1})$ and $\bbe_{j+1}=(\be_0,\be_1,\dots,\be_{j+1})$.
These operators commute with each other, i.e.
$$[\cL_j(z;\be),\cL_k(z;\be)]=0,$$
see \cite[Section 3]{GI}.
If we think of $z_{k+1}$ and $\be$ as parameters and we consider the space of polynomials $\Rset[w_1,\dots,w_k]$, where 
\begin{equation}\label{4.7}
w_s=w_s(z;\be)=z_s(z_s+\be_s),
\end{equation}
then using the $I$-invariance one can show that 
\begin{equation*}
\cL_j(z;\be):\Rset[w_1,\dots,w_k]\to \Rset[w_1,\dots,w_k] \qquad\text{ for }j=1,2,\dots,k.
\end{equation*}
Moreover, the operators $\cL_j(z;\be)$, $j=1,\dots,k$ can be simultaneously diagonalized on $\Rset[w_1,\dots,w_k]$ by the multivariable Racah polynomials defined by Tratnik in \cite{Tr}.
Explicitly, if we define for $\nu\in\Nset_0^{k}$ polynomials by 
\begin{align*}
  &R_k(\nu;z;\beta)=   \prod_{j=1}^k (2 |\bnu_{j-1}| +\be_j-\be_0)_{\nu_j} 
   ( |\bnu_{j-1}|+\be_{j+1}+z_{j+1})_{\nu_j}( |\bnu_{j-1}|-z_{j+1})_{\nu_j} \\
      & \qquad \quad \times  
       {}_4F_3 \left (\begin{matrix} -\nu_j, \nu_j+  2 |\bnu_{j-1}| +\be_{j+1}-\be_0-1,   |\bnu_{j-1}|-x_j,   |\bnu_{j-1}| +\be_j+z_j\\
      2 |\bnu_{j-1}| +\be_j-\be_0, |\bnu_{j-1}|+\be_{j+1}+z_{j+1},  |\bnu_{j-1}|-z_{j+1} \end{matrix}; 1 \right),
\end{align*}
where $\nu_0=0$, then one can show that $R_k(\nu;z;\beta)\in \Rset[w_1,\dots,w_k]$ and 
\begin{equation}\label{4.8}
\cL_j(z;\be)R_k(\nu;z;\beta)=\la_j(|\bnu_j|;\be)R_k(\nu;z;\beta), \text{ for }j=1,\dots,k,
\end{equation}
where
\begin{equation}\label{4.9}
\la_j(s;\be)=-s(s+\be_{j+1}-\be_{0}-1), \text{ for }j=1,\dots,k,
\end{equation}
see \cite[Theorem~3.9]{GI}. If $z_{k+1}=N\in\Nset$, then we consider the above polynomials for $|\nu|\leq N$ and they are mutually orthogonal on the set $V_{k+1,N}=\{z\in\Nset_0^{k+1}:0\le z_1\leq z_2\leq \cdots\leq z_k\leq z_{k+1}=N\}$ with respect to the weight
\begin{equation}\label{4.10}
\rho_k(z;\beta)= \prod_{j=0}^k \frac{(\be_{j+1}-\be_j)_{z_{j+1}-z_j}(\be_{j+1})_{z_{j+1}+z_j}}
{(z_{j+1}-z_j)! (\be_{j}+1)_{z_{j+1}+z_j}} \prod_{j=1}^k \frac{((\be_j+2)/2)_{z_j}}{(\be_j/2)_{z_j}}
\end{equation}
with norms given by 
\begin{align*}
&||R_k(\nu;\cdot;\beta)||^2= \sum_{z\in V_{k+1,N}} \rho_k(z;\beta)R^2_{k}(\nu;z;\beta) \\
&=\frac{(\be_{k+1})_{N+|\nu|} (-N)_{|\nu|}(-N-\be_0)_{|\nu|} (2|\nu|+\be_{k+1}-\be_{0})_{N-|\nu|}}{N!\,(\be_0+1)_N}\\ 
     & \times 
     \prod_{j=1}^{k} \nu_j!(\be_{j+1}-\be_{j})_{\nu_j}(2|\bnu_{j-1}|+\be_{j}-\be_{0})_{\nu_j}(|\bnu_{j}|+|\bnu_{j-1}|+\be_{j+1}-\be_{0}-1)_{\nu_j}.
\end{align*}
In this case, the orthonormal polynomials are defined by
$$\Rh_k(\nu;z;\beta)=\frac{1}{||R_k(\nu;\cdot;\beta)||} R_k(\nu;z;\beta) \qquad \text{ for } |\nu|\leq N.$$

Finally, we denote by 
\begin{equation}\label{4.11}
\cR_k(z;\be)=\Rset\langle w_1(z;\be),\dots,w_k(z;\be), \cL_1(z;\be),\dots,\cL_k(z;\be) \rangle,
\end{equation}
the associative algebra over $\Rset$ generated by the operators in \eqref{4.6} and the variables $w_j$ in \eqref{4.7} and we will refer to it as the multivariable Racah algebra.

\section{Symmetry algebra and Racah operators}\label{se5}

Note that the action of the symmetric group $S_{d+1}$ on $\td$ defined by $\tau (t_{i,j})=t_{\tau_i,\tau_j}$ for $\tau\in S_{d+1}$ corresponds to the simultaneous permutation of the variables $x=(x_1,\dots,x_{d+1})$ and the parameters $\ga=(\ga_1,\dots,\ga_{d+1})$, where $x_{d+1}=1-|x|$. Moreover, the Dirichlet distribution in \eqref{3.1} is invariant under this action.

In the rest of the paper we fix $\tau\in S_{d+1}$ to be the cyclic permutation
\begin{equation}\label{5.1}
\tau=(1,2,\dots,d,d+1),
\end{equation}
and we denote by 
\begin{equation}\label{5.2}
\fGd^{\tau}=\tau \circ \fGd
\end{equation}
the Gaudin subalgebra of $\td$ obtained by applying $\tau$ to the Gaudin subalgebra $\fGd$. Thus $\fGd^{\tau}$ is generated by $\cM^{\tau}_j=\tau\circ \cM_j$, $j=1,2\dots,d$. 
Next, let
\begin{equation}\label{5.3}
P^{\tau}_{\nu}(x;\ga)=\tau\circ P_{\nu}(x;\ga)=P_{\nu}(\tau\circ x;\tau\circ\ga)
\end{equation}
denote the orthogonal polynomials with respect to the Dirichlet distribution, obtained by applying $\tau$ to the orthogonal basis defined in \eqref{3.4}.
Finally, let 
$$\Ph_{\nu}(x;\ga)=\frac{1}{||P_{\nu}||}P_{\nu}(x;\ga)\qquad\text{ and }\qquad\Ph^{\tau}_{\nu}(x;\ga)=\frac{1}{||P^{\tau}_{\nu}||}P^{\tau}_{\nu}(x;\ga),$$ 
denote the orthonormal polynomials with respect to the inner product \eqref{3.2}, obtained by normalizing the polynomials $P_{\nu}(x;\ga)$ and $P^{\tau}_{\nu}(x;\ga)$, respectively. For $n\in\Nset_0$, the transition matrix between these two orthonormal bases of $\cP^{\ga}_n$ can be written explicitly in terms of the multivariable Racah polynomials. More precisely, define 
\begin{subequations}\label{5.4}
\begin{equation}\label{5.4a}
\beh_{j}=\ga_1+|\bga^{d+2-j}|+j\qquad\text{ for }j=0,1,\dots,d,
\end{equation}
and for $\nu,\mu \in \Nset_0^d$ let
\begin{equation}\label{5.4b}
\nuh=(|\bnu^d|,|\bnu^{d-1}|,\dots,|\bnu^{2}|,|\bnu^1|)\qquad \text{ and } \qquad 
\mub=(\mu_d,\mu_{d-1},\dots,\mu_2).
\end{equation}
\end{subequations}
\begin{subequations}\label{5.5}
For $\mu \in \Nset_0^d$ let
\begin{equation}\label{5.5a}
\mut=(|\bmu_1|,|\bmu_{2}|,\dots,|\bmu_{d-1}|,|\bmu_{d}|),
\end{equation}
and for $n\in\Nset_0$ we  define
\begin{equation}\label{5.5b}
\bet_0=\bet_0(n)=\ga_1, \qquad \bet_{j}(n)=-|\bga^{j+1}|-2n-d+j\qquad\text{ for }j=1,\dots,d.
\end{equation}
\end{subequations}
Then for $|\nu| = |\mu|=n$ we have
\begin{align}
\langle \Ph^{\tau}_{\nu}(x;\ga), \Ph_{\mu}(x;\ga)\rangle 
&=(-1)^n \sqrt{\rho_{d-1}(\nuh; \beh)}\,  \Rh_{d-1}(\mub;\nuh; \beh)\label{5.6}\\
&=(-1)^n\sqrt{\rho_{d-1}(\mut; \bet(n))}\,  \Rh_{d-1}(\bnu_{d-1};\mut; \bet(n)),\label{5.7}
\end{align}
see Section 6 in \cite{IX}.

\begin{Proposition}\label{pr5.1} 
For $\nu,\mu\in\Nset_0^{d}$ and $j=2,3,\dots,d$ we have 
\begin{align}
&\cM_{j,d}(x) P^{\tau}_{\nu}(x;\ga)=\cL_{d+1-j}(\nuh;\beh)P^{\tau}_{\nu}(x;\ga)\label{5.8}\\
&\cM^{\tau}_{j,d}(x) P_{\mu}(x;\ga)\nonumber\\
&\quad=\left(|\mu|(|\mu|+\bet_{j}(|\mu|)-\bet_0-1)+\frac{1}{g_{d}(\mu;\ga)}\cL_{j-1}(\mut;\bet(|\mu|))\circ g_{d}(\mu;\ga)\right)P_{\mu}(x;\ga),\label{5.9}
\end{align}
where  $(\beh,\nuh)$ and $(\bet,\mut)$ are defined in equations \eqref{5.4} and \eqref{5.5}, respectively, and
\begin{equation}\label{5.10}
g_{d}(\mu;\ga)=\frac{(1+\ga_1)_{\mu_1}}{(|\ga|+2|\mu|+d-\mu_1)_{\mu_1}}.
\end{equation}
\end{Proposition}

\begin{Remark}\label{re5.2}
The operator $\cL_{d+1-j}(\nuh;\beh)$ in \eqref{5.8} is a difference operator in the variables $\nu_{1},\dots,\nu_{d}$ obtained from the operator in \eqref{4.6} by changing the variables. Explicitly, we replace $z_l$ by $|\bnu^{d+1-l}|$ for $l=1,2,\dots,d$ in the coefficients and we replace $E_{z_l}$ by $E_{\nu_{d+1-l}}E_{\nu_{d-l}}^{-1}$ for $l=1,2,\dots,d-1$. The operator $\cL_{j-1}(\mut;\bet(|\mu|))$ in equation \eqref{5.9} is defined in a similar manner.
\end{Remark}

\begin{proof}[Proof of \prref{pr5.1}]
From \coref{co3.2} we know that $\cM_{j,d}(x) P^{\tau}_{\nu}(x;\ga)\in\cP^{\ga}_{|\nu|}$ and therefore 
\begin{equation}\label{5.11}
\cM_{j,d}(x) P^{\tau}_{\nu}(x;\ga)=L_{j}(\nu)P^{\tau}_{\nu}(x;\ga)
\end{equation}
for some difference operator $L_j(\nu)=\sum_{|s|=0}l_s(\nu;\ga)E_{\nu}^s$ acting on the indices $\nu$, with coefficients depending on $\nu$ and $\ga$. Thus, equation \eqref{5.11} is equivalent to the equations
\begin{equation}\label{5.12}
\langle \cM_{j,d}(x) P^{\tau}_{\nu}(x;\ga), \Ph_{\mu}(x;\ga)\rangle=\langle  L_{j}(\nu)P^{\tau}_{\nu}(x;\ga), \Ph_{\mu}(x;\ga) \rangle
\end{equation}
for all $\mu\in\Nset_0^{d}$ such that $|\mu|=|\nu|$.

Using \eqref{3.5} and  \eqref{4.10} one can deduce that $||P^{\tau}_{\nu}||^2\rho_{d-1}(\hat \nu; \beh)=1$.
Combining this with Propositions \ref{pr3.1}, \ref{pr3.3} and equations \eqref{4.9}, \eqref{5.4}, \eqref{5.6} we see that the left-hand side of \eqref{5.12} can be written as follows
\begin{align}
&\langle \cM_{j,d}(x) P^{\tau}_{\nu}(x;\ga), \Ph_{\mu}(x;\ga)\rangle=\langle P^{\tau}_{\nu}(x;\ga),  \cM_{j,d}(x)\Ph_{\mu}(x;\ga)\rangle \nonumber\\
&\qquad=-|\bmu^{j}|(|\bmu^{j}|+|\bga^{j}|+d+1-j)\langle P^{\tau}_{\nu}(x;\ga), \Ph_{\mu}(x;\ga)\rangle \nonumber\\
&\qquad=-|\bar{\bmu}_{d+1-j}|(|\bar{\bmu}_{d+1-j}|+\beh_{d+2-j}-\beh_0-1)||P^{\tau}_{\nu}|| \langle \Ph^{\tau}_{\nu}(x;\ga), \Ph_{\mu}(x;\ga)\rangle \nonumber\\
&\qquad=(-1)^{|\mu|}\la_{d+1-j}(|\bar{\bmu}_{d+1-j}|;\beh) ||P^{\tau}_{\nu}|| \sqrt{\rho_{d-1}(\hat \nu; \beh)}\,  \Rh_{d-1}(\mub;\nuh; \beh) \nonumber\\
&\qquad=(-1)^{|\mu|}\la_{d+1-j}(|\bar{\bmu}_{d+1-j}|;\beh)  \Rh_{d-1}(\mub;\nuh; \beh).\label{5.13}
\end{align}
For the right-hand side of  \eqref{5.12} we obtain
\begin{align}
\langle  L_{j}(\nu)P^{\tau}_{\nu}(x;\ga), \Ph_{\mu}(x;\ga) \rangle %\nonumber\\
%&\qquad
&= L_{j}(\nu) \langle  P^{\tau}_{\nu}(x;\ga), \Ph_{\mu}(x;\ga) \rangle \nonumber\\
&= L_{j}(\nu) ||P^{\tau}_{\nu}|| \langle \Ph^{\tau}_{\nu}(x;\ga), \Ph_{\mu}(x;\ga)\rangle \nonumber\\
&=(-1)^{|\mu|} L_{j}(\nu) ||P^{\tau}_{\nu}|| \sqrt{\rho_{d-1}(\hat \nu; \beh)}\,  \Rh_{d-1}(\mub;\nuh; \beh) \nonumber\\
&= (-1)^{|\mu|}L_{j}(\nu)  \Rh_{d-1}(\mub;\nuh; \beh).\label{5.14}
\end{align}
From equations \eqref{4.8}, \eqref{5.13} and \eqref{5.14} it is clear that the operator $L_{j}(\nu)$ must coincide with the Racah operator $\cL_{d+1-j}(\nuh;\beh)$, completing the proof of equation \eqref{5.8}. The proof of \eqref{5.9} follows along the same lines, using \eqref{5.7}.
\end{proof}

Summarizing all statements so far, we can formulate the main result of the paper, which gives explicit formulas for the action of the Gaudin algebras $\fGd$ and $\fGd^{\tau}$ on each of the bases $\{P_{\mu}(x;\ga):\mu\in\Nset_0^d, \;|\mu|=n\}$ and $\{P^{\tau}_{\nu}(x;\ga):\nu\in\Nset_0^d, \;|\nu|=n\}$ of $\cP^{\ga}_n$ in terms of the multivariable Racah algebra $\cR_{d-1}$ defined in \eqref{4.11}.

\begin{Theorem}\label{th5.3}
Let $n\in\Nset$. For $\mu\in\Nset_0^d$ such that $|\mu|=n$ and for $j\in\{2,\dots,d\}$ we have
\begin{subequations}\label{5.15}
\begin{align}
&\cM_{1,d}(x)P_{\mu}(x;\ga) = \cM_{1,d}^{\tau}(x)P_{\mu}(x;\ga)=-n(n+|\ga|+d)P_{\mu}(x;\ga), \label{5.15a}\\
&\cM_{j,d}(x)P_{\mu}(x;\ga) = \left(n(n+\bet_{j-1}(n))-w_{j-1}(\mut,\bet(n))\right) P_{\mu}(x;\ga),\label{5.15b}\\
&\cM^{\tau}_{j,d}(x) P_{\mu}(x;\ga)\nonumber\\
&\quad=\left(n(n+\bet_{j}(n)-\bet_0-1)+\frac{1}{g_{d}(\mu;\ga)}\cL_{j-1}(\mut;\bet(n))\circ g_{d}(\mu;\ga)\right)P_{\mu}(x;\ga),\label{5.15c}
\end{align}
\end{subequations}
where  $(\bet,\mut)$ and $g_d$ are defined in equations \eqref{5.5} and \eqref{5.10}, respectively. Likewise, for $\nu\in\Nset_0^d$ such that $|\nu|=n$ and for $j\in\{2,\dots,d\}$ we have
\begin{subequations}\label{5.16}
\begin{align}
&\cM^{\tau}_{1,d}(x)P^{\tau}_{\nu}(x;\ga) = \cM_{1,d}(x)P^{\tau}_{\nu}(x;\ga)=-n(n+|\ga|+d)P^{\tau}_{\nu}(x;\ga), \label{5.16a}\\
&\cM^{\tau}_{j,d}(x)P^{\tau}_{\nu}(x;\ga) = - w_{d+1-j}(\nuh;\beh) P^{\tau}_{\nu}(x;\ga),\label{5.16b}\\
&\cM_{j,d}(x) P^{\tau}_{\nu}(x;\ga)=\cL_{d+1-j}(\nuh;\beh)P^{\tau}_{\nu}(x;\ga),\label{5.16c}
\end{align}
\end{subequations}
where $\beh$ and $\nuh$ are defined in equations \eqref{5.4}.
\end{Theorem}

\begin{Remark}\label{re5.4}
When $d=2$, $\tau=(1,2,3)$ and we have $\cM_1=\cM_1^{\tau}=t_{1,2}+t_{1,3}+t_{2,3}$, $\cM_2=t_{2,3}$, $\cM^{\tau}_{2}=t_{1,3}$. Equivalently, we have $t_{1,2}=\cM_1-\cM_2-\cM^{\tau}_{2}$, $t_{1,3}=\cM^{\tau}_{2}$, $t_{2,3}=\cM_2$ and therefore the formulas in the above theorem give explicit formulas for the action of all elements of $\mathfrak{t}_3$ on the basis $\{P_{\mu}(x;\ga):\mu\in\Nset_0^d, \;|\mu|=n\}$ of $\cP^{\ga}_n$.

When $d=3$, $\tau=(1,2,3,4)$ and we can use equations \eqref{5.15} and \eqref{5.16c} to express the action of $\cM_{j}$, $\cM^{\tau}_j$ and $\cM^{\tau^{-1}}_j=\tau^{-1}\circ \cM_{j}$ for all $j\in\{1,2,3\}$ on the basis $\{P_{\mu}(x;\ga):\mu\in\Nset_0^d,\;|\mu|=n\}$ of $\cP^{\ga}_n$. Again, it is not hard to see that we can take appropriate linear combinations of these elements to obtain $t_{i,j}$ for all $1\leq i<j\leq 4$ and therefore we can write explicit formulas for the action of all elements of  $\mathfrak{t}_4$ in terms of the Racah operators using \thref{th5.3}. 

When $d>3$, the elements $\cM_j$, $\cM^{\tau}_j$ and $\cM^{\tau^{-1}}_j$ still generate $\td$ and therefore  \thref{th5.3} describes the action of all elements, but we need to use the nonlinear relation \eqref{2.3}.
\end{Remark}

\begin{Theorem}\label{th5.5}
We have
\begin{align}
t_{1,j}&=(\cM^{\tau}_{j-1}-\cM_{j})-(\cM^{\tau}_{j}-\cM_{j+1}), \text{ for }j=2,\dots,d+1,\label{5.17}\\
t_{i,d+1}&=(\cM_{i}-\cM^{\tau^{-1}}_{i+1})-(\cM_{i+1}-\cM^{\tau^{-1}}_{i+2}), \text{ for }i=1,\dots,d,\label{5.18}
\end{align}
with the convention that $\cM_{d+1}=\cM_{d+2}=0$. Moreover, the set
\begin{equation}\label{5.19}
\cSt=\{\cM_j:j=1,2,\dots,d\}\cup \{\cM^{\tau}_j:j=2,\dots,d\}\cup  \{\cM^{\tau^{-1}}_j:j=2,\dots,d\}
\end{equation} 
generates $\td$.
\end{Theorem}

\begin{proof}
From the definition of $\cM_j$ in equation \eqref{3.7} it is easy to see that
\begin{equation*}
\cM_j^{\tau}=\cM_{j+1}+\sum_{k=j+1}^{d+1}t_{1,k},
\end{equation*}
which gives \eqref{5.17}. Equation \eqref{5.18} follows by applying $\tau^{-1}$ to \eqref{5.17}. Equations \eqref{5.17} and \eqref{5.18} imply that the algebra generated by $\cSt$ contains the elements $\cS$ in \eqref{2.5} and therefore
the proof that $\cSt$ generates $\td$ follows from \prref{pr2.3}.
\end{proof}

Theorems \ref{th5.3} and \ref{5.5} together with equation \eqref{2.3} give explicit formulas for the action of all elements of $\td$ on the basis $\{P_{\mu}(x;\ga):\mu\in\Nset_0^d, \;|\mu|=n\}$ of $\cP^{\ga}_n$.

\begin{Remark}[Connection to bispectrality]\label{re5.6}
Fix $k\in\Nset$, $z_{k+1}=N\in\Nset$ and consider the Racah polynomials $R_k(\nu;z;\be)$ defined in the previous section. Besides the difference equations \eqref{4.8} in the variables $z$, they satisfy also difference equations in the indices $\nu$. More precisely, we can construct a second family $\{\cB_j(\nu;\be)\}_{j=1,2,\dots,k}$ of commuting partial difference operators in $\nu$, which are independent of $z_1,\dots,z_{k}$, such that 
\begin{equation}\label{5.20}
\cB_j(\nu;\be)R_k(\nu;z;\beta)=\ka_j(z;\be)R_k(\nu;z;\beta), \text{ for }j=1,\dots,k,
\end{equation}
where the eigenvalues $\ka_j(z;\be)$ are independent of $\nu$, see Section 4 in \cite{GI} for details.
In view of the work of Duistermaat and Gr\"unbaum \cite{DG}, we refer to equations \eqref{4.8} and \eqref{5.20} as {\em bispectral equations\/} for the Racah polynomials. 
Note that equations~\eqref{5.8} are essentially equivalent to the spectral equations \eqref{4.8}, upon using \eqref{5.6}. More precisely, the spectral equations \eqref{4.8} with $k=d-1$ are equivalent to equations \eqref{5.8} if we use the fact that $\fGd$ acts diagonally on the basis  $\{\hat{P}_{\mu}(x;\ga)\}$ combined with the identity 
\begin{equation}\label{5.21}
\langle \cN \hat{P}^{\tau}_{\nu}(x;\ga), \Ph_{\mu}(x;\ga)\rangle=\langle \hat{P}^{\tau}_{\nu}(x;\ga),  \cN\Ph_{\mu}(x;\ga) \rangle \qquad \text{ for }\cN\in\fGd,
\end{equation}
and formula \eqref{5.6}. Similarly, we can derive the spectral equations \eqref{5.20} for $k=d-1$ by using equation \eqref{5.9}, the fact that $\fGd^{\tau}$ acts diagonally on the basis  $\{\hat{P}^{\tau}_{\nu}(x;\ga)\}$, the identity 
\begin{equation}\label{5.22}
\langle \cN^{\tau} \hat{P}^{\tau}_{\nu}(x;\ga), \Ph_{\mu}(x;\ga)\rangle=\langle \hat{P}^{\tau}_{\nu}(x;\ga),  \cN^{\tau}\Ph_{\mu}(x;\ga) \rangle \qquad \text{ for }\cN^{\tau}\in\fGd^{\tau},
\end{equation}
and formula \eqref{5.6}. Therefore, the bispectral algebras of difference operators in $z$ and $\nu$ are parametrized by the Gaudin subalgebras $\fGd$ and $\fGd^{\tau}$, respectively. There is an interesting parallel between the present constructions and the ones in \cite{I1}, where bispectral commutative algebras of partial difference operators were constructed for multivariable polynomials, orthogonal with respect to the multinomial distribution. The key ingredients there were specific representations of the Lie algebra $\sld$ and two Cartan subalgebras which parametrize the corresponding bispectral commutative algebras of partial difference operators, while here we use representations of the Kohno-Drinfeld Lie algebra with two Gaudin subalgebras. 
It would be interesting to extend the above results to the bispectral commutative algebras constructed in \cite{I2} for the multivariable $q$-Racah polynomials defined in \cite{GR}, and to relate them to an appropriate quantum integrable system. 
\end{Remark}

Finally, we note that while this paper was under review, an interesting link between the theory developed here and the Laplace-Dunkl operator for $\Zset_2^{n}$ appeared in \cite{DGVV}.

\end{document}